\documentclass[12pt]{l4dc2021}

\usepackage{tikz}
\usepackage{pgfplots}
\usepackage{txfonts}         

\title[Offset-free setpoint tracking using neural network controllers]{Offset-free setpoint tracking using neural network controllers}
\usepackage{times}
\author{%
 \Name{Patricia Pauli} \Email{patricia.pauli@ist.uni-stuttgart.de}\\
 \addr Institute for Systems Theory and Automatic Control, University of Stuttgart%
 \AND
 \Name{Johannes Köhler} \Email{jkoehle@ethz.ch}\\
 \addr Institute for Dynamic Systems and Control, ETH Zurich%
 \AND
 \Name{Julian Berberich} \Email{julian.berberich@ist.uni-stuttgart.de}\\
 \Name{Anne Koch} \Email{anne.koch@ist.uni-stuttgart.de}\\
 \Name{Frank Allgöwer} \Email{frank.allgower@ist.uni-stuttgart.de}\\
 \addr Institute for Systems Theory and Automatic Control, University of Stuttgart%
}

\begin{document}

\maketitle

\begin{abstract}%
In this paper, we present a method to analyze local and global stability in offset-free setpoint tracking using neural network controllers and we provide ellipsoidal inner approximations of the corresponding region of attraction. We consider a feedback interconnection of a linear plant in connection with a neural network controller and an integrator, which allows for offset-free tracking of a desired piecewise constant reference that enters the controller as an external input. %The feedback interconnection considered in this paper allows for general configurations of the neural network controller that include the special cases of output error and state feedback.
Exploiting the fact that activation functions used in neural networks are slope-restricted, we derive linear matrix inequalities to verify stability using Lyapunov theory. After stating a global stability result, we present less conservative local stability conditions (i) for a given reference and (ii) for any reference from a certain set. The latter result even enables guaranteed tracking under setpoint changes using a reference governor which can lead to a significant increase of the region of attraction. Finally, we demonstrate the applicability of our analysis by verifying stability and offset-free tracking of a neural network controller that was trained to stabilize a linearized inverted pendulum.
\end{abstract}


\section{Introduction}
Alongside with the flourishing success of deep learning methods, in recent times, there has been an increasing interest in feedback interconnections with neural network (NN) controllers. NNs as universal function approximators can learn any desired, nonlinear behavior with high precision and there are efficient ways to train them, e.g., using backpropagation and deep reinforcement learning (RL) \citep{sutton2018reinforcement}. NNs are useful whenever no precise model is available and in case it is expensive to obtain a model but data can be collected easily \citep{levin1993control, psaltis1988multilayered, ibrahim2020learning}. NN controllers that are determined from supervised learning can approximate more complicated controllers and, in comparison to the original controller, they might save computation time \citep{hertneck2018learning}. Another practically highly relevant application of NN controllers is in terms of deep RL, e.g., using inverse reinforcement learning where RL is used to imitate human behavior \citep{yin2020imitation}. However, even though NN controllers may work well in practice, they lack safety guarantees \citep{szegedy2013intriguing} which is the reason why they are typically not used in safety-critical applications such as autonomous driving.
%For instance, the underlying optimization problem for training of the NN controller may state a more general cost function (e.g., not just LQR).
%Moreover, looking at the bigger picture, feedback interconnections of LTI systems with NN controllers may be considered a starting point to guarantee closed-loop properties for NN controllers in feedback interconnection with nonlinear systems.}  

Recently, efforts have been made to give robustness and stability guarantees for NNs and feedback interconnections with NN controllers. \cite{fazlyab2019efficient} certify Lipschitz bounds that serve as a proxy for robustness of NNs and \cite{pauli2020training, revay2020lipschitz} impose them during training. Alongside with robustness guarantees, there are recent works on stability guarantees, e.g., \cite{revay2020contracting, yin2020stability} present results on stability of recurrent neural networks using incremental $\mathcal{L}_2$-bounds and feedback interconnections using Lyapunov theory, respectively. The underlying property that is exploited in the works stated above is \textit{slope-restriction}, i.e., the slopes of activation functions comply with certain bounds, which can be formulated as an incremental quadratic constraint. Lipschitz continuity and closed-loop stability, respectively, are then verified using semidefinite programming.

In this paper, we extend the stability analysis of a linear system in feedback interconnection with an NN controller in~\citep{yin2020stability} to offset-free setpoint tracking with a piecewise constant reference.  In comparison to simple linear feedback controllers, NN controllers may provide more flexibility, e.g., they may be used to achieve additional goals beside the mere stabilization of the closed loop using problem-specific objective functions. Note that \cite{yin2020stability} require a stability analysis tailored to any possible steady state. In contrast, we are able to derive stability results for a whole set of references, considering a feedback interconnection with an external input and an NN controller in combination with an integrator feedback, that avoids steady-state offsets. %\cite{yin2020stability} study closed-loop stability with NN controllers. Beside stability, we additionally pursue the control goal of offset-free setpoint tracking and for that purpose, we use a feedback interconnection with an external input, an integrator feedback, that avoids steady-state offsets, and an NN controller.
%Whereas the stability analysis in \citep{yin2020stability} is restricted to state feedback, we choose a more general description of the NN controller which includes the important special cases of output error feedback and state feedback. Especially, output error feedback is desirable as it requires the output measurement only, not the full state or a costly estimate of the state. Yet, state information may still be used during training of the NN to improve performance, cf. \citep{zhang2016learning}.
We show that via online modification of the reference, a reference governor ensures stability in an extended region of attraction (RoA). %similar to~\cite{enforcing2020donti}, where instead of the reference, the control input is adjusted by projection onto a set that guarantees exponential stability.
In this paper, we give three main theorems that verify stability based on linear matrix inequalities (LMIs). First, we discuss global stability of the suggested feedback interconnection using a Lyapunov function. We proceed with a local stability condition for a fixed reference which is less restrictive than the global stability result, also analyzing the resulting RoA. Finally, we formulate a local stability condition for all references from a specified range and for this case, we show how a reference governor can be used to increase the RoA. 

The paper is organized as follows: In Section \ref{sec:setpoint_tracking}, we present the problem setting and formulate ingredients relevant to the stability analysis. Following, in Section \ref{sec:stability}, we state LMI conditions for global stability and local stability (i) for a fixed constant reference and (ii) for any reference from a certain set which facilitates tracking of piecewise constant references from that set, respectively. In addition, we present an approach to extend the guaranteed RoA using a reference governor. In Section \ref{sec:example}, we provide a numerical example and we summarize our results in Section \ref{sec:conclusion}.

\section{Setpoint tracking}\label{sec:setpoint_tracking}
In this section, we state the problem setting and formulate all ingredients to verify stability in offset-free setpoint tracking using NN controllers. We use $c\in[a,b]$ for vectors $c$ that comply componentwise with the range $[a,b]$ and $|c|$ denotes the componentwise magnitude of the entries of~$c$.

\subsection{Problem setting}
We consider discrete-time linear time-invariant systems
\begin{equation*}\label{eq:G}
\begin{split}
\text{G:}\qquad x_{k+1}&=Ax_k+Bu_k\\
y_k&=Cx_k
\end{split}
\end{equation*}
with state $x_k\in\mathbb{R}^{n_x}$, input $u_k\in\mathbb{R}^{n_u}$, and output $y_k\in\mathbb{R}^{n_r}$, of the same input and output dimension ($n_u=n_r$). 
The control goal is twofold: We aim for (i) stability  of the closed-loop system and (ii)~offset-free tracking of a constant reference~$r$, i.e., $y_k\to r$ for $k\to\infty$.

The considered controller architecture in feedback interconnection with the plant $\text{G}$ is illustrated in Fig.~\ref{fig:Output_feedback}, where the reference $r\in\mathbb{R}^{n_r}$ enters as an external input. The controller $\text{K}$ consists of two parts, (i) an NN controller and (ii) an integrator. In order to track a constant reference~$r$, we use an $l$-layer feed-forward neural network $\kappa(x,r)$ as a controller whose input-output behavior $\kappa: \mathbb{R}^{n_x}\times\mathbb{R}^{n_r}\to \mathbb{R}^{n_r}$ reads
%\begin{minipage}{0.5\textwidth}
\begin{equation*}
\begin{split}
w^0_k&=H^0_xx_k+H^0_rr\\
\text{NN}:\qquad w^{i+1}_k&=\phi^{i+1}(W^iw^i_k+b^i)\\
\kappa(x_k,r)&=W^lw^l_k+b^l,%+H^l \begin{bmatrix}x_k \\ r\end{bmatrix},
\end{split}
\end{equation*}
where $w_k^i\in\mathbb{R}^{n_i}$ are the outputs of the neurons, $W^i\in\mathbb{R}^{n_{i+1}\times n_{i}}$ are weights, $b^i\in\mathbb{R}^{n_{i+1}}$ are biases, and $\phi^{i}:\mathbb{R}^{n_i}\to \mathbb{R}^{n_i}$ is the vector of activation functions of the $i$-th layer for all $i=1,...,l$ hidden layers. The state $x_k$ and the reference $r$ enter the NN controller through the matrices $H_x^0$ and $H_r^0$. While \cite{yin2020stability} consider only state feedback, our description of the NN controller is more general, including the two important special cases of state feedback, i.e., $w^0=x$, $H^0_x=I_{n_x},~H^0_r=0_{n_x\times n_r}$, and output error feedback, i.e., $w^0=r-y$, $H^0_x=-C,~H^0_r=I_{n_r\times n_r}$. The case of output error feedback is especially relevant as it requires the output measurement only, not the full state or a costly estimate thereof. Yet, state information may still be used during training of the NN to improve performance, cf. \citep{zhang2016learning}.
By adding an integrator feedback to the controller with dynamics 
 \begin{equation*}
\xi_{k+1}=\xi_k+r-y_{k},
\end{equation*} 
introducing the additional state $\xi_k$, an undesired steady-state offset is eliminated independent of the structure of the weights and biases of the NN. The resulting input to the plant $\text{G}$ then is $u_k=k_{\xi}\xi_k+\kappa(x_k,r)$. Here, $k_\xi\in\mathbb{R}^{n_r\times n_r}$ is a constant gain that is assumed to be invertible which is required for offset-free tracking. 
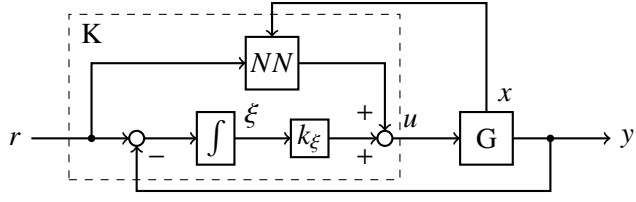
\begin{figure}
\centering
		\begin{tikzpicture}
			\draw[->,thick] (0.3,0) node[left]{$r$} -- (1.6,0);	
			\filldraw[fill=white, thick] (1.7,0) node[below right]{$-$} circle (0.1cm);
			\fill (1.1,0) circle (0.05cm);
			\draw[->,thick] (1.1,0) -- (1.1,1) -- (3.15,1);
			\draw[->,thick] (1.8,0) -- (2.5,0);
			\path (3.5,1) node {$NN$};
			\draw[thick] (3.15,0.65) rectangle (3.85,1.35);
			\path (2.75,0) node {$\int$};
			\draw[thick] (2.5,-0.35) rectangle (3,0.35);
			\path (4,0) node {$k_\xi$};
			\draw[thick] (3.75,0.25) rectangle (4.25,-0.25);
			\draw[->,thick] (5.1,0) node[above right ]{$u$} -- (6,0);	
			\filldraw[fill=white, thick] (5,0) node[below left]{$+$} circle (0.1cm);
			\draw[->,thick] (3.85,1) -- (5,1) -- (5,0.1) node[above left]{$+$} ;
			\draw[->,thick] (3,0) node[above right ]{$\xi$} -- (3.75,0);
			\draw[->,thick] (4.25,0)  -- (4.9,0);
			\path (6.35,0) node {G};
			\draw[thick] (6,-0.35) rectangle (6.7,0.35);
			\draw[->,thick] (6.35,0.35) node[above right]{$x$} -- (6.35,1.8) -- (3.5,1.8)  -- (3.5,1.35);
			\draw[->,thick] (6.7,0) -- (8,0) node[right] {$y$};
			\fill (7.2,0) circle (0.05cm);
			\draw[->,thick] (7.2,0) -- (7.2,-0.7) -- (1.7,-0.7)  -- (1.7,-0.1);
			% \draw[->, thick, dashed] (7,-1) node[right]{$\eta$} -- (6.1,-1);
			% \filldraw[fill=white, thick, dashed] (6,-1) circle (0.1cm);
			\draw[dashed] (0.8,-0.55) rectangle (5.2,1.65);
			\path (1.1,1.4) node {K};
		\end{tikzpicture}
		\vspace{-0.2cm}
\caption{Feedback interconnection of plant $\text{G}$ and controller $\text{K}$, consisting of NN controller and integrator feedback $k_\xi$.}
\vspace{-0.4cm}
\label{fig:Output_feedback}
\end{figure}

Introducing the augmented state $\tilde{x}=\begin{bmatrix}x^\top & \xi^\top\end{bmatrix}^\top$ yields the overall dynamics of the augmented system
\begin{equation}\label{eq:G_tilde}
\begin{split}
\widetilde{\text{G}}:\qquad \tilde{x}_{k+1}&=\tilde{A}\tilde{x}_k+\tilde{B}\tilde{u}_k+B_rr\\
y_k&=\tilde{C}\tilde{x}_k
\end{split}
\end{equation}
with $\tilde{u}=\kappa(x,r)$ and the reference $r$ as inputs, where
\begin{equation*}
\tilde{A}=\begin{bmatrix}A & Bk_{\xi}\\ -C & I_{n_r}\end{bmatrix},~\tilde{B}=\begin{bmatrix}B\\0\end{bmatrix},~\tilde{C}=\begin{bmatrix}C & 0\end{bmatrix},~B_r=\begin{bmatrix}0\\I_{n_r}\end{bmatrix}.
\end{equation*}

For a given constant reference $r$ that is tracked with zero steady-state offset according to the defined control goal, the corresponding state and input steady state $(x_*,u_*)$ satifies
\begin{equation}\label{eq:steady_state}
\begin{bmatrix}
0\\r
\end{bmatrix}
=
\underbrace{\begin{bmatrix}
A-I & B\\
C & 0
\end{bmatrix}}_{\eqqcolon A_a}
\begin{bmatrix}
x_*\\u_*
\end{bmatrix}.
\end{equation}
Throughout this paper, we assume that the matrix $A_a$ is square and non-singular, which is a standard condition in tracking \citep{pannocchia2003disturbance}. For any given reference $r$, this leads to a unique steady state $(x_*, u_*)$  that achieves the desired output $r=Cx^*$. With the information of the steady state $(x_*, u_*)$, we can determine the corresponding unique steady state $\xi_*=k_\xi^{-1}(u_*-\kappa(x_*,r))$ and the stationary value $w^i_*$ from a forward propagation of $x_*$ and $r$ through the NN.

With the corresponding steady state $(\tilde{x}_*,\tilde{u}_*)$ of \eqref{eq:G_tilde}, we formulate the incremental dynamics
\begin{equation*}
\begin{split}
\tilde{x}_{k+1}-\tilde{x}_*&=\tilde{A}(\tilde{x}_k-\tilde{x}_*)+\tilde{B}(\tilde{u}_k-\tilde{u}_*),\\
y_k-r&=\tilde{C}(\tilde{x}_k-\tilde{x}_*).
\end{split}
\end{equation*}
Stated in error coordinates $e_k:=\tilde{x}_k-\tilde{x}_*$, the plant dynamics are independent of the reference $r$, which will be exploited in the analysis. %Note that $\tilde{u}=\kappa(x,r)$ may depend on the reference $r$, a fact the incremental formulation used in the analysis accounts for, as well.
The goal of offset-free setpoint tracking is achieved if $e_k\to0$ for $k\to\infty$.

\subsection{Slope restriction}
Like it was done by \cite{fazlyab2019efficient} and many subsequent works in the field, e.g., in \citep{pauli2020training, revay2020contracting, yin2020stability}, we exploit that activation functions in NNs are slope-restricted. This is an incremental sector condition on the individual activation functions of each of the neurons of the NN. We define $v^{i+1}=W^iw^i+b^i$ such that $w^i=\phi(v^i)$ and denote the input to the $j$-th neuron of the $i$-th layer at time $k$ by $v_k^{i,j}$ and the steady-state input to the same neuron by $v_*^{i,j}$. Throughout this paper, we consider nonlinear activation functions $\varphi:\mathbb{R}\to\mathbb{R}$ that are slope-restricted on the interval $[\alpha,\beta]$, i.e.,
\begin{equation}\label{eq:sector_bounds}
\alpha \leq \frac{\varphi(v^{i,j})-\varphi(v^{i,j}_*)}{v^{i,j}-v^{i,j}_*} \leq \beta \quad \forall i\in 1,\dots,l\quad\forall j\in 1,\dots,n_i.
\end{equation}
with $0\leq\alpha<\beta$. This applies to the most common activation functions, e.g., to ReLU and $\tanh$ with $\alpha=0,~\beta=1$. The property of slope-restriction can then be stated as an incremental quadratic constraint
\begin{equation*}\label{eq:QC_NN_1_neuron}
\begin{bmatrix}
v^{i,j}-v^{i,j}_*\\
\varphi(v^{i,j})-\varphi(v^{i,j}_*)
\end{bmatrix}^\top
\begin{bmatrix}
-2\alpha\beta & (\alpha+\beta)\\
(\alpha+\beta) & -2
\end{bmatrix}
\begin{bmatrix}
v^{i,j}-v^{i,j}_*\\
\varphi(v^{i,j})-\varphi(v^{i,j}_*)
\end{bmatrix}\geq 0.
\end{equation*}
We stack up all activations $\phi=[{\phi^1}^\top,\cdots,{\phi^l}^\top]^\top:\mathbb{R}^n\to\mathbb{R}^n$ with $n=\sum_{i=1}^l n_i$ as well as $v=[{v^1}^\top,\cdots,{v^l}^\top]^\top$ and $w=[{w^1}^\top,\cdots,{w^l}^\top]^\top$ and introduce multipliers $\lambda_{i}\geq0$ to formulate the incremental quadratic constraint for all neurons
\begin{equation}\label{eq:QC_NN}
\begin{bmatrix}
v-v_*\\
\phi(v)-\phi(v_*)
\end{bmatrix}^\top
%\underbrace{
\begin{bmatrix}
-2\alpha\beta \Lambda & (\alpha+\beta)\Lambda\\
(\alpha+\beta)\Lambda & -2\Lambda
\end{bmatrix}%}_{:=M(\Lambda)}
\begin{bmatrix}
v-v_*\\
\phi(v)-\phi(v_*)
\end{bmatrix}\geq 0
\end{equation}
with the diagonal matrix $\Lambda\in\mathbb{D}_+^n\coloneqq\left\{\lambda_i\in \mathbb{R}\mid\Lambda=\mathrm{diag}(\lambda_1,\dots,\lambda_n), \lambda_i\geq0\right\}$. %:=\{\Lambda\in\mathbb{R}^n\mid \Lambda=\lambda_{ii} e_i {e_i}^\top, \lambda_{ii}\geq0\}$.
\cite{fazlyab2019efficient} claimed that \eqref{eq:QC_NN} holds for a richer class of multipliers $\Lambda$. However, as shown by counterexample in \citep{pauli2020training}, this is not true and the restriction to diagonal matrices $\Lambda$ is indeed necessary.
    
\subsection{Basis transformations}
For convenience, we use a compact notation similar to \cite{yin2020stability} by defining
\vspace{-0.2cm}
\begin{equation}\label{eq:basis_transformations}
N_0\coloneqq\begin{bmatrix} N_0^{1} \\ 0\end{bmatrix}=\begin{bmatrix} W^0H^0_x & 0 \\ 0 & 0 \end{bmatrix},~N_{1:l-1}\coloneqq\begin{bmatrix} 0 & \cdots & 0 & 0\\ W^1 & \dots & 0 & 0\\ 0 & \ddots &\vdots  & \vdots\\ 0 & \cdots & W^{l-1} & 0 \end{bmatrix},~N_l\coloneqq\begin{bmatrix} 0 & W^l \end{bmatrix}.
\end{equation}
Using these definitions, the mapping of the NN controller is characterized by 
\begin{equation*}
\begin{bmatrix}
\tilde{x}-\tilde{x}_*\\
\tilde{u}-\tilde{u}_*
\end{bmatrix}=
\underbrace{\begin{bmatrix}
I_{n_{\tilde{x}}} & 0_{n_{\tilde{x}}\times n}  \\
0_{n_u\times n_{\tilde{x}}} & N_l
\end{bmatrix}}_{:=R_V}
\begin{bmatrix}
{\tilde{x}}-{\tilde{x}}_*\\
w-w_*
\end{bmatrix},~
\begin{bmatrix}
v-v_*\\
w-w_*
\end{bmatrix}=
\underbrace{\begin{bmatrix}
N_0 & N_{1:l-1}  \\
0_{n\times n_{\tilde{x}}} & I_{n}
\end{bmatrix}}_{:=R_{\phi}}
\begin{bmatrix}
{\tilde{x}}-{\tilde{x}}_*\\
w-w_*
\end{bmatrix}.
\end{equation*}

\section{Stability analysis}\label{sec:stability}
In this section, we state our main results to verify stability in offset-free setpoint tracking. First, we provide a theorem to show global stability and then state two local stability results with inner approximations of the corresponding RoA. Finally, we show how an extended RoA is accessible via a reference governor.
\subsection{Global stability}
In this section, we state an LMI condition that verifies global stability of the feedback interconnection of $\widetilde{\text{G}}$ and $\kappa(x,r)$, with reference $r$ as an external input. Global stability indicates that the system converges to its steady state and achieves offset-free tracking for all references and all initial conditions.
\begin{theorem}\label{th:global}
Suppose there exist matrices $\Lambda\in\mathbb{D}_+^n$, $P\succ0$, such that
\begin{equation}\label{eq:LMI}
R_V^\top
\begin{bmatrix}
\tilde{A}^\top P \tilde{A}-P & \tilde{A}^\top P \tilde{B}\\
\tilde{B}^\top P \tilde{A} & \tilde{B}^\top P \tilde{B}
\end{bmatrix}
R_V+
R_{\phi}^\top
\begin{bmatrix}
-2\alpha\beta \Lambda & (\alpha+\beta)\Lambda\\
(\alpha+\beta)\Lambda & -2\Lambda
\end{bmatrix}
R_{\phi}\prec0 
\end{equation}
holds.
Then, for any initial condition $\tilde{x}_0\in \mathbb{R}^{n_{\tilde{x}}}$ and for any reference $r\in \mathbb{R}^{n_r}$, the NN controller $\kappa(x,r)$ ensures exponential stability of the steady state $\tilde{x}_*$ for the closed-loop system~\eqref{eq:G_tilde} and achieves zero steady-state offset, i.e., $y_k\to r$ for $k\to\infty$.
\end{theorem}
\begin{proof}
We use the fact that the incremental quadratic constraint \eqref{eq:QC_NN} holds. Given that \eqref{eq:LMI} is strict, there exists an $\epsilon>0$, such that right/left multiplying \eqref{eq:LMI} with $[(\tilde{x}_k-\tilde{x}_*)^\top~(w_k-w_*)^\top]^\top$ yields
\begin{equation}\label{eq:LMI_pf}
\begin{bmatrix}
\tilde{x}_k-\tilde{x}_*\\
\tilde{u}_k-\tilde{u}_* 
\end{bmatrix}^\top
\begin{bmatrix}
\tilde{A}^\top P \tilde{A}-P & \tilde{A}^\top P \tilde{B}\\
\tilde{B}^\top P \tilde{A} & \tilde{B}^\top P \tilde{B}
\end{bmatrix}
\begin{bmatrix}
\tilde{x}_k-\tilde{x}_*\\
\tilde{u}_k-\tilde{u}_*
\end{bmatrix}\leq -\epsilon\|\tilde{x}_k-\tilde{x}_*\|^2.
\end{equation}
We define the Lyapunov function $V(\tilde{x}):=(\tilde{x}-\tilde{x}_*)^\top P (\tilde{x}-\tilde{x}_*)$ and translate \eqref{eq:LMI_pf} into $V(\tilde{x}_{k+1})-V(\tilde{x}_k)\leq-\epsilon \| \tilde{x}_k-x_*\|^2$ which shows global exponential stability.
\end{proof}

While this result can verify global stability of the considered feedback interconnection, in practice, NN controllers are usually trained on training data from a problem-relevant region resulting from physical constraints of the environment. Oftentimes, NN controllers hence perform well at and around the training data set, yet show unpredictable behavior on unknown regions. For that reason, (i) it is unlikely to have a globally stabilizing NN controller and (ii) mostly local stability properties corresponding to the relevant data region are sufficient in practice. Imposing local stability is less restrictive and hence allows for a better performance. Therefore, in the next subsection, we present a result to certify local stability.

\begin{remark}
Although the focus of this paper is on analysis, the resulting LMI \eqref{eq:LMI} can be incorporated in the training of NN controllers analogous to \cite{pauli2020training,revay2020convex}.
\end{remark}

\subsection{Local stability for a fixed reference}\label{sec:local_1}
As considering local sector bounds that are tighter than the global ones can reduce conservatism, we formulate a local stability condition for the feedback interconnection of $\widetilde{\text{G}}$ and $\kappa(x_k,r)$, with a fixed constant reference $r$ as an external input. In \citep{yin2020stability}, a similar local analysis was performed for the reference-free case.  First, in this subsection, we consider local stability for a given reference~$r$ and in the next subsection, we provide a local stability result for references from a set.

For a given constant reference $r$, we choose some $d\in\mathbb{R}^{n_1}, d_i>0$ and form a symmetrical set around the steady state $v^1_*$ of the inputs to the neurons of the first hidden layer $v^1\in[v^1_*-d,v^1_*+d]=[\underline{v}^1,\overline{v}^1]$. With the bounds on $v^1$, we can explicitly determine the bounds on all entries of $v\in[\underline{v},\overline{v}]$ and $w\in[\underline{w},\overline{w}]$. Similar to \cite{yin2020stability}, we define the local incremental sector bounds $\alpha_\phi^{i,j}$ and $\beta_\phi^{i,j}$ of the $j$-th neuron of the $i$-th layer of the NN such that the local slope-restriction property
\begin{equation*}
\alpha_\phi^{i,j} \leq \frac{\varphi(v^{i,j})-\varphi(v^{i,j}_*)}{v^{i,j}-v^{i,j}_*} \leq \beta_\phi^{i,j} \quad \forall i\in 1,\dots,l\quad\forall j\in 1,\dots,n_i
\end{equation*}
holds for all elements of $v\in[\underline{v},\overline{v}]$ and such that the bounds are tight. Moreover, we define the diagonal matrices
\begin{align*}
\alpha_\phi^i\coloneqq\textrm{diag}(\alpha_\phi^{i,j}),~j=1,\dots,n_i,\quad\alpha_\phi\coloneqq\textrm{blkdiag}(\alpha_\phi^i),~j=1,\dots,l,\\
\beta_\phi^i\coloneqq\textrm{diag}(\beta_\phi^{i,j}),~j=1,\dots,n_i,\quad\beta_\phi\coloneqq\textrm{blkdiag}(\beta_\phi^i),~j=1,\dots,l,
\end{align*}
the ellipsoidal set $\mathcal{E}_P(\tilde{x}_*):=\left\{\tilde{x}\in\mathbb{R}^{n_x}\mid(\tilde{x}-\tilde{x}_*)^\top P(\tilde{x}-\tilde{x}_*)\leq 1\right\}$, and we have $v^1-v^1_*=N_0^1(\tilde{x}-\tilde{x}_*)$ with $N_0^1$ as defined in \eqref{eq:basis_transformations}.
The index $j$ in $N_{0,j}^1$ and $d_j$ denotes the $j$-th row of $N_0^1$ and respectively, the $j$-th entry of $d$. With these definitions, we state an LMI condition for local stability.
\begin{theorem}\label{th:local_1}
Suppose there exist matrices $\Lambda\in\mathbb{D}_+^n$, $P\succ0$ such that
\begin{subequations}\label{eq:local_1}
\begin{equation}\label{eq:LMI_local}
R_V^\top
\begin{bmatrix}
\tilde{A}^\top P \tilde{A}-P & \tilde{A}^\top P \tilde{B}\\
\tilde{B}^\top P \tilde{A} & \tilde{B}^\top P \tilde{B}
\end{bmatrix}
R_V+
R_x^\top
\begin{bmatrix}
-2\alpha_\phi\beta_\phi \Lambda & (\alpha_\phi+\beta_\phi)\Lambda\\
(\alpha_\phi+\beta_\phi)\Lambda & -2\Lambda
\end{bmatrix}
R_x\prec0,
\end{equation}
\begin{equation}\label{eq:LMI_RoA}
\begin{bmatrix}
d^2_j  & N_{0,j}^1\\
{N_{0,j}^1}^\top & P
\end{bmatrix}
\succeq0, \qquad \forall j=1,\dots, n_1
\end{equation}
\end{subequations}
holds. Then, for any initial condition $\tilde{x}_0\in \mathcal{E}_P(x_*)$ the NN controller $\kappa(x,r)$ ensures exponential stability of the steady state $\tilde{x}_*$ for the closed-loop system \eqref{eq:G_tilde} and achieves zero steady-state offset, i.e., $y_k\to r$ for $k\to\infty$.
\end{theorem}
\begin{proof}
The proof is a direct extension of Theorem~\ref{th:global} using the derivation in~\cite[Thm.~1]{yin2020stability}.
\end{proof}

In order to maximize the RoA, instead of solving solely the feasibility problem \eqref{eq:local_1}, we can solve the semidefinite program with the objective $\min_{P,\Lambda} \mathrm{trace}{P}$ and LMI constraints \eqref{eq:local_1}.
Even though Theorem~\ref{th:local_1} verifies local stability and therefore, as we argued above, it is potentially less conservative than Theorem~\ref{th:global}, this result still is somewhat limiting due to the fact that it is stated for a fixed reference~$r$. Tracking a different constant reference in general necessitates to recalculate the incremental sector bounds $\alpha_\phi$ and $\beta_\phi$. For example, for state feedback, $v^1_*$ changes with $r$ and it is required to newly determine the bounds on $v$ and $w$ as well as the incremental sector bounds $\alpha_\phi$ and $\beta_\phi$ for a different reference $r$. Finally, the LMI conditions \eqref{eq:local_1} need to be checked again with the new values of $\alpha_\phi$ and $\beta_\phi$. For the special case of output error feedback, $v_*$ conveniently remains the same for all references~$r$, given that $v^1_*=W^0(r-Cx_*)+b^0=b^0~\forall r$ with \eqref{eq:steady_state}, such that $\alpha_\phi$ and $\beta_\phi$ remain the same for different references~$r$. This means in the case of output error feedback, if Theorem~\ref{th:local_1} holds, the controller $\kappa(x,r)$ ensures local exponential stability of the closed-loop system~\eqref{eq:G_tilde} for all references $r\in\mathbb{R}^{n_r}$ and all intial conditions $\tilde{x}_0\in\mathcal{E}_P(\tilde{x}_*)$, i.e., an ellipsoidal set whose size is independent of the reference $r$ and that is centered at the corresponding setpoint $\tilde{x}_*$. This means that for error output feedback local stability depends only on the tracking error.

\subsection{Local stability for a range of references}
In this subsection, we state a more general local stability result that guarantees stability for all references from a certain set.
According to \eqref{eq:steady_state}, the steady-state manifold is characterized by 
\begin{equation*}
\begin{bmatrix}
x_*\\u_*
\end{bmatrix}=
A_a^{-1}
\begin{bmatrix}
0\\r
\end{bmatrix}=
\begin{bmatrix}
* & M\\
* & *
\end{bmatrix}
\begin{bmatrix}
0\\r
\end{bmatrix}.
\end{equation*}
Thus, we obtain a linear map $x_*(r)=Mr$. Using this map and the fact that we use an offset-free tracking formulation (due to the integrator), we can characterize the set of setpoints  $x_*(r)$ for all admissible references. %$r\in[\underline{r},\overline{r}]$.
A forward pass of $x_*(r)$ through the NN gives the values of $v_*(r)$ and $w_*(r)$, as well. Carrying out the stability analysis for a range of references necessitates to determine the incremental sector bounds $\alpha_\phi$ and $\beta_\phi$ according to Subsection \ref{sec:local_1} for a given nominal reference $r_\textrm{nom}\in\mathbb{R}^{n_r}$ and a given vector $d\in\mathbb{R}^{n_1},~d_i>0$.
This time, we aim to find an underestimate of the RoA formulated in the augmented state $\tilde{x}$ and the desired reference $r$ and therefore, we introduce the ellipsoidal set
\begin{equation*}
\mathcal{E}_{P,Q}(r_\textrm{nom}):=\left\{\tilde{x}\in\mathbb{R}^{n_{\tilde{x}}},~r\in\mathbb{R}^{n_r}\mid(\tilde{x}-\tilde{x}_*(r))^\top P(\tilde{x}-\tilde{x}_*(r))+(r-r_\textrm{nom})^\top Q(r-r_\textrm{nom})\leq 1\right\}.
\end{equation*} 
This choice is natural as it addresses the trade-off between the initial conditions $\tilde{x}_0$ and the choice of the constant reference $r$. For $r=r_\textrm{nom}$, we retrieve the ellipsoidal set $\mathcal{E}_P(\tilde{x}_*)$. The ellipse in $\tilde{x}$ corresponding to $r=r_\textrm{nom}$ is the largest possible for the given vector $d$, whereas for references $r$ that deviate from the nominal reference $r_\textrm{nom}$, the corresponding ellipsoidal set in $\tilde{x}$ shrinks: $(\tilde{x}-\tilde{x}_*(r))^\top P(\tilde{x}-\tilde{x}_*(r))\leq 1-(r-r_\textrm{nom})^\top Q(r-r_\textrm{nom})$. Finally, the union of the ellipses for all $r\in\mathcal{E}_{P,Q}(r_\textrm{nom})$ provides an enlarged guaranteed RoA if the reference is adjusted online, cf. Subsection \ref{sec:reference_governor}.

\begin{theorem}\label{th:local_2}
Suppose there exist matrices $\Lambda\in\mathbb{D}_+^n$, $P\succ0$, $Q\succ0$ such that
\begin{subequations}
\begin{equation}\label{eq:LMI_local_2}
R_V^\top
\begin{bmatrix}
\tilde{A}^\top P \tilde{A}-P & \tilde{A}^\top P \tilde{B}\\
\tilde{B}^\top P \tilde{A} & \tilde{B}^\top P \tilde{B}
\end{bmatrix}
R_V+
R_x^\top
\begin{bmatrix}
-2\alpha_\phi\beta_\phi \Lambda & (\alpha_\phi+\beta_\phi)\Lambda\\
(\alpha_\phi+\beta_\phi)\Lambda & -2\Lambda
\end{bmatrix}
R_x\prec0,
\end{equation}
\begin{equation}\label{eq:LMI_RoA_2}
\begin{bmatrix}
d^2_j  & \begin{bmatrix}N_0^1 & N_0^1M\end{bmatrix}_j\\
\begin{bmatrix}N_0^1 & N_0^1M\end{bmatrix}_j^\top & \begin{bmatrix} P & 0 \\ 0 & Q \end{bmatrix}
\end{bmatrix}
\succeq 0, \qquad \forall j=1,\dots, n_1
\end{equation}
\end{subequations}
holds. Then, for all pairs of initial conditions and references $(\tilde{x}_0,r) \in\mathcal{E}_{P,Q}(r_\textrm{nom})$ the NN controller $\kappa(x,r)$ ensures exponential stability of the steady state $\tilde{x}_*(r)$ for the closed-loop system \eqref{eq:G_tilde} and achieves zero steady-state offset, i.e., $y_k\to r$ for $k\to\infty$.
\end{theorem}
\begin{proof}
First, we show that for all $(\tilde{x},r)\in\mathcal{E}_{P,Q}(r_\textrm{nom})$, the input to the first neuron $v^1$ stays between the given bounds, i.e., $v^1\in[\underline{v}^1,\overline{v}^1]$. %Then, $\mathcal{E}_{P,Q}(r_\textrm{nom})$ is an inner approximation of the RoA. 
Applying the Schur complement to LMI \eqref{eq:LMI_RoA_2} yields 
\begin{equation}\label{eq:Schur}
\begin{bmatrix}N_0^1 &  N_0^1M\end{bmatrix}_j \begin{bmatrix} P^{-1} & 0 \\ 0 & Q^{-1}\end{bmatrix}\begin{bmatrix}N_0^1 &  N_0^1M\end{bmatrix}_j^\top\leq d_j^2,\quad \forall j=1,\dots,n_1.
\end{equation}
According to Lemma 1 in \citep{hindi1998analysis}, from \eqref{eq:Schur} we get 
\begin{equation*}
\mathcal{E}_{P,Q}(r_\textrm{nom})\subseteq \left\{\tilde{x}\in\mathbb{R}^{n_{\tilde{x}}},~r\in\mathbb{R}^{n_r} \mid \left\vert\begin{bmatrix}N_0^1 &  N_0^1M\end{bmatrix}_j\begin{bmatrix}\tilde{x}-\tilde{x}_*(r) \\ r-r_\textrm{nom}\end{bmatrix}\right\vert\leq d_j,~j=1\dots,n_1\right\}
\end{equation*}
and
\begin{equation*}
\left\vert\begin{bmatrix}N_0^1 &  N_0^1M\end{bmatrix}\begin{bmatrix}\tilde{x}-\tilde{x}_*(r) \\ r-r_\textrm{nom}\end{bmatrix}\right\vert
%=\left\vert W^0(x-x_*(r_\textrm{nom})\rigt\vert
=\left\vert v^1-v_*^1(r_\textrm{nom})\right\vert\leq d,
\end{equation*}
from which we conclude
$\mathcal{E}_{P,Q}(r_\textrm{nom})\subseteq \left\{\tilde{x}\in\mathbb{R}^{n_{\tilde{x}}},~r\in\mathbb{R}^{n_r} \mid v^1\in[\underline{v}^1,\overline{v}^1]\right\}$.
Hence, the given positive invariance of $\mathcal{E}_{P,Q}(r_\textrm{nom})$ ensures invariance of the bounds $[\underline{v},\overline{v}]$. Closed-loop stability and offset-free tracking then follow from \eqref{eq:LMI_local_2} analogous to the proof of Theorem~\ref{th:global} with $\alpha_\phi$ and $\beta_\phi$ instead of $\alpha$ and $\beta$, and $\mathcal{E}_{P,Q}(r_\textrm{nom})$ is an inner approximation of the RoA. 
\end{proof}

With this result, we can verify stability locally for all combinations of initial conditions $\tilde{x}_0$ and references~$r$ from the ellipsoidal set~$\mathcal{E}_{P,Q}(r_\textrm{nom})$ and thus, it allows for tracking of piecewise constant references. In the next subsection, we present an approach to increase the guaranteed RoA.
\begin{remark}
Model uncertainties, e.g., formulated as integral quadratic constraints, can be incorporated into the feedback interconnection considered in this paper to analyze robust stability and robust estimates of the RoA, cf. \citep{yin2020stability}.
\end{remark}

\subsection{Extended guaranteed RoA via a reference governor}\label{sec:reference_governor}
A reference governor is an augmentation to a locally stabilizing controller that is often employed to ensure satisfaction of state and input constraints. In particular, a reference governor uses the current state measurement to modify the reference that is fed into the controller~\citep{bemporad1998reference}. 

In the given setup of this paper, a reference governor provides a tool to extend the RoA. After computing the matrices $P$ and $Q$ via Theorem~\ref{th:local_2}, at all times $k$, the reference governor checks if the combination of the current state $\tilde{x}_k$ and the reference $r$ lie in the guaranteed RoA $\mathcal{E}_{P,Q}$. If $(\tilde{x}_k, r)\notin\mathcal{E}_{P,Q}$, e.g., due to a change in the reference, then the reference governor adjusts the reference to the closest possible reference $\hat{r}_k$ for which the current state $\tilde{x}_k$ lies in the RoA and the surrogate reference $\hat{r}_k$ is fed into the closed-loop system \eqref{eq:G_tilde}. In this way, even if the desired reference leads to an unstable closed loop, the reference governor ensures stability.
A corresponding surrogate reference $\hat{r}_k$ can be computed according to the following formulation:
\begin{equation}\label{eq:RG}
\hat{r}_k\coloneqq\arg\min_{\hat{r}} \left\|r-\hat{r}\right\|^2~\text{s.t.}~(\tilde{x}_k-\tilde{x}_*(\hat{r}))^\top P(\tilde{x}_k-\tilde{x}_*(\hat{r}))+(\hat{r}-r_\textrm{nom})^\top Q(\hat{r}-r_\textrm{nom})\leq 1.
\end{equation}
The surrogate reference $\hat{r}_k$ is updated at every time step until the current state lies in the guaranteed RoA corresponding to the desired reference $r$ (in case this is possible) such that $\hat{r}_k=r$. Even if the reference $r$ is outside of the guaranteed RoA, the reference closest to the desired $r$ that lies in the estimate of the RoA can be tracked. Overall, the reference governor provides stability independent of the desired reference $r$ and guarantees desirable closed-loop properties for any $\tilde{x}_0$ in the union of the local ellipsoidal sets $\mathcal{E}_{P,Q}$ over all admissible references $r$, cf. Fig.~\ref{fig:ellipses}. Note that this reference governor has similarities to the approach in \citep{enforcing2020donti}, where instead of the reference, the control input is adjusted by projection onto a set that guarantees exponential stability.
\begin{remark}
Note that in the special case of output error feedback, local stability is independent of $r$ and depends only on the tracking error. Hence, a corresponding reference governor is given by 
\begin{equation*}
\hat{r}_k=\arg\min_{\hat{r}} \left\|r-\hat{r}\right\|~\textrm{s.t.}~(\tilde{x}_k-\tilde{x}_*(\hat{r}))^\top P(\tilde{x}_k-\tilde{x}_*(\hat{r}))\leq 1.
\end{equation*}
Thus, in the output error feedback case, we can achieve offset-free tracking for any reference $r\in \mathbb{R}^{n_r}$, if the reference governor is feasible for the initial state.
\end{remark}

\section{Numerical example}\label{sec:example}
In this section, we illustrate our results on an example. The analysis is applicable to all NN controllers within the suggested control architecture, regardless of how they were trained. In the following example, an NN controller approximates a model predictive controller (MPC) which is beneficial in case the MPC is computationally expensive.

Consider the problem of steering an inverted pendulum with mass $m=0.15\,\textrm{kg}$, length $L=0.5\,\textrm{m}$, and friction coefficient $\mu=0.5\,\textrm{Nm/rad}$ to a desired reference position. The linearized dynamics are
\begin{equation}\label{eq:pendulum}
\dot{x}=\begin{bmatrix}
0 & 1\\ \frac{g}{L} & -\frac{\mu}{mL^2}
\end{bmatrix}
x+
\begin{bmatrix}
0 \\ \frac{g}{L}
\end{bmatrix}u
,\quad
y=\begin{bmatrix}0 & 1\end{bmatrix}x
\end{equation}
with state $x=\begin{bmatrix}\theta^\top & \dot{\theta}^\top\end{bmatrix}^\top$, angular position $\theta$, angular velocity  $\dot{\theta}$, and control input $u$.
We use supervised learning to determine a stabilizing controller for the dynamics \eqref{eq:pendulum} discretized with sampling time $T_s=0.02\,\textrm{s}$. For this purpose, we train an NN controller such that it approximates an MPC trained on the discretized and linearized dynamics of the inverted pendulum with input constraint $|u|\leq 1$, similar to~\cite{zhang2016learning}. For training of the NN, we utilize closed-loop trajectories with initial conditions sampled from a uniform distribution $x_0\in\mathcal{U}(-0.5,0.5)$ and we set $k_\xi=1$. We consider an NN, with activation function $\tanh$ and two layers of 5 neurons each, that serves as a state feedback controller, i.e., $w_0=x$. While the approach in \citep{yin2020stability} can only handle zero bias to ensure the setpoint is $x_*=0,~u_*=0$, in our setup no structural constraints on the weights and biases of the NN are necessary, as any undesired offset is eliminated by the integrator.

Theorem~\ref{th:global} is not applicable here, meaning that we cannot show global stability of the closed loop. This is not surprising, given that the NN controller was trained on trajectories from a problem-relevant region. To compute a local RoA, we choose $d_i=0.345,~i=1,\dots,n_1$, and then propagate the resulting bounds $[\underline{v}^1,\overline{v}^1]$ through the NN to obtain the incremental sector bounds $\alpha_\phi$ and $\beta_\phi$. With vector $d$ and the resulting bounds $\alpha_\phi$ and $\beta_\phi$, we apply Theorem~\ref{th:local_1} with the reference $r=0$ (computation time $1.14$s) and Theorem~\ref{th:local_2} with the nominal reference $r_\textrm{nom}=0$ (computation time $0.498$s) using numerical SDP solvers \citep{Lofberg2004, mosek} in Matlab.

Fig. \ref{fig:ellipses} shows the corresponding ellipses $\mathcal{E}_P(\tilde{x}_*(0))$ (blue, solid) and $\mathcal{E}_{P,Q}(0)$ (red, dashed) for different $r\in\mathcal{E}_{P,Q}(0)$, as well as the steady-state manifold $x_*(r)$ of the problem. Based on the result of Theorem~\ref{th:local_2}, we can track any reference angle $r\in[-0.2,0.2]$. Further, note that in this example, the extended RoA, i.e., the union of the ellipses $\mathcal{E}_{P,Q}(0)$, is notably larger than $\mathcal{E}_P(\tilde{x}_*(0))$. We also implemented a reference governor to track the reference $r=-1$. Whenever the constraint in \eqref{eq:RG} is not fulfilled with the current state and the desired reference $r$, the reference governor determines a surrogate reference $\hat{r}_k$ closest to the desired reference. 
In Fig.~\ref{fig:ellipses}, a trajectory $x_{RG}$ for a random initial condition under control with the reference governor is shown. This trajectory tracks the closest point to $r$ within the extended RoA with zero steady-state offset while $x_{RG}$ stays in the extended guaranteed RoA at all times. Note that without the reference governor the trajectory diverges.

\begin{figure}\label{fig:ellipses}
\centering
\input{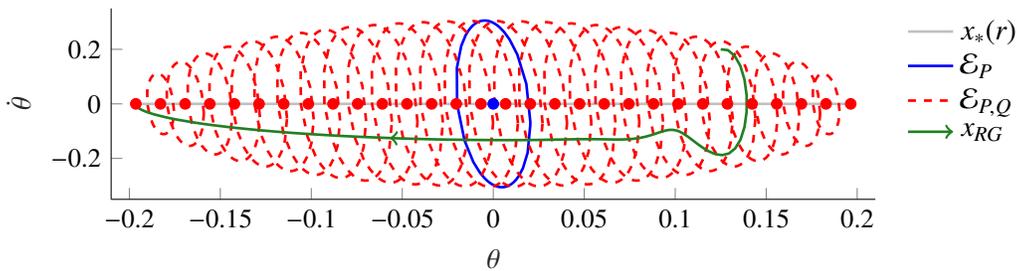}
\vspace{-0.4cm}
\caption{Steady-state manifold $x_*(r)$, ellipsoidal inner approximations of RoA $\mathcal{E}_P$ and $\mathcal{E}_{P,Q}$, and example trajectory $x_{RG}$ using a reference governor for the linearized inverted pendulum.}
\vspace{-0.6cm}
\end{figure}

\section{Conclusion}\label{sec:conclusion}
We proposed a method to analyze local and global stability in offset-free setpoint tracking using NN controllers, also providing ellipsoidal inner approximations of the corresponding RoA.% We chose a feed-forward neural network as a controller that allows for a general parametrization of the NN controller, which includes the special cases of output error and state feedback.
We derived linear matrix inequalities to verify stability using Lyapunov theory, exploiting the fact that activation functions used in NNs are slope-restricted. First, we showed global stability and then stated less restrictive local stability results (i) for a given reference and (ii) for a set of references, which allows for tracking of piecewise constant references. Based on the latter result, we were able to significantly increase the guaranteed RoA using a reference governor. Finally, we tested our stability analysis on the example of a linearized inverted pendulum.

An interesting direction for future research is to extend the guaranteed RoA by optimizing over the range $[\underline{v}^1,\overline{v}^1]$ of the input to the first layer, that in the current setup is chosen by the user.

%%Acknowledgments---Will not appear in anonymized version
\acks{This work was funded by Deutsche Forschungsgemeinschaft (DFG,
German Research Foundation) under Germany’s Excellence Strategy - EXC
2075 - 390740016. The authors thank the International Max Planck Research
School for Intelligent Systems (IMPRS-IS) for supporting Patricia Pauli, Julian Berberich, and Anne Koch.}% <-this % stops a space

\bibliography{references}

\end{document}